\documentclass[]{IEEEtran}

\IEEEoverridecommandlockouts

\usepackage{amsmath} 
\usepackage{amssymb}  
\usepackage{mathtools}
\usepackage{amsthm}
\usepackage{xcolor}
\usepackage{tikz}
\usetikzlibrary{automata, positioning, arrows}
\usepackage{empheq}

\newtheorem{theorem}{Theorem}
\newtheorem{proposition}{Proposition}

\theoremstyle{definition}
\newtheorem{definition}{Definition}

\theoremstyle{definition}
\newtheorem{assumption}{Assumption}

\theoremstyle{definition}

\theoremstyle{definition}
\newtheorem{problem}{Problem}

\theoremstyle{remark}
\newtheorem{remark}{Remark}

\title{\LARGE \bf
Model-Based Reinforcement Learning for Approximate Optimal Control with Temporal Logic Specifications
}

\author{Max H. Cohen and Calin Belta%
\thanks{The authors are with the Department of Mechanical Engineering, Boston University, Boston, MA \texttt{\{maxcohen,cbelta\}@bu.}}
}

\begin{document}

\maketitle
\thispagestyle{empty}
\pagestyle{empty}

\begin{abstract}
In this paper we study the problem of synthesizing optimal control policies for uncertain continuous-time nonlinear systems from syntactically co-safe linear temporal logic (scLTL) formulas. We formulate this problem as a sequence of reach-avoid optimal control sub-problems.  We show that the resulting hybrid optimal control policy guarantees the satisfaction of a given scLTL formula by constructing a barrier certificate. Since solving each optimal control problem may be computationally intractable, we take a learning-based approach to approximately solve this sequence of optimal control problems online without requiring full knowledge of the system dynamics. Using Lyapunov-based tools, we develop sufficient conditions under which our approximate solution maintains correctness. Finally, we demonstrate the efficacy of the developed method with a numerical example.
\end{abstract}

\section{Introduction}
The intersection of hybrid systems \cite{Teel_book}, formal methods \cite{ModelChecking_book}, and control theory provides a pathway towards controlling dynamical systems to achieve objectives richer than stability \cite{Belta_book}. These objectives can range from simple reachability and safety specifications \cite{TomlinAutomatica99,TomlinTAC05} to more complex specifications such as those expressed as temporal logic (TL) formulas. Classical approaches to designing controllers for dynamical systems subject to TL specifications involve \emph{abstraction-based} techniques \cite{Belta_book,Tabuada_book}. In this setting, the dynamical system is abstracted into a finite transition system (FTS) whereas the specification, often represented as a linear TL (LTL) formula, is captured by a finite-state automaton (FSA). This FTS and FSA can then be combined to form a product automaton and the problem of synthesizing a correct (with respect to the specification) control policy for the original system is reduced to solving a B\"{u}chi or Rabin game over the product automaton. These approaches have been extensively studied in the literature \cite{BeltaTAC08,TabuadaTAC06}; however, a well-documented limitation is the computational burden associated with the abstraction process.

To address the limitations of abstraction-based control design, more recent work has focused on \emph{optimization-based} control with TL constraints \cite{BeltaARCRAS19}. Here, TLs with quantitative semantics, such as signal temporal logic (STL) \cite{STL}, are leveraged to form optimization problems with the goal of maximizing the satisfaction of the desired specification. These approaches often rely on encoding Boolean satisfaction of STL formulas as mixed-integer constraints, leading to mixed-integer linear/quadratic programs (MILP/MIQP) that can be solved within model predictive control (MPC) frameworks \cite{MurrayHSCC15,SadraAllerton15}. Although these approaches avoid the computationally prohibitive abstraction process, solving MILPs/MIQPs in a receding horizon fashion is still computationally intensive and may prevent real-time implementation. To address this limitation, other approaches have explored using control barrier functions (CBFs) \cite{AmesTAC17} to encode subsets of LTL or STL formulas within more computationally efficient quadratic programs \cite{CooganCDC18,LarsLCSS19}. 

Recently, optimization-based control from TL specifications has also been explored within more traditional Hamilton-Jacobi-Bellman (HJB)-based optimal control frameworks \cite{TopcuCDC16,TopcuHSCC17,VamvoudakisACC20}. Similar to abstraction-based approaches, these methods rely on translating a TL formula into an FSA but circumvent the computational expense of abstracting the dynamical system. The authors in \cite{TopcuCDC16,TopcuHSCC17} augment a continuous dynamical system with the discrete behavior of an FSA corresponding to a syntactically cosafe LTL (scLTL) formula to form a hybrid dynamical system. The authors then solve a mixed continuous-discrete-time HJB equation corresponding to an optimal control problem to obtain a correct-by-construction hybrid feedback controller. For linear systems, this problem can be solved through semi-definite programming; however, for nonlinear systems this same approach results in a semi-infinite program, which is difficult to solve in general. Other recent work in this domain includes \cite{VamvoudakisACC20}, where the authors propose an approximate dynamic programming (ADP) approach to solve a sequence of two-point boundary value problems, whose solution yields a control policy capable of satisfying an scLTL formula.

With the exception of accounting for bounded external disturbances, the aforementioned approaches rely on exact model knowledge. In practice, however, models often exhibit some degree of uncertainty, which motivates the use of \emph{learning-based} control methods that can handle uncertainty while preserving formal guarantees. One way to address this challenge is to take a completely model-free approach as in \cite{CooganCDC14,BeltaCDC16}; however, it is generally difficult to maintain formal guarantees in a model-free setting \cite{SadraHSCC18}. An alternative approach is to take a model-based approach wherein an approximate model of the system is identified, which is then used to establish formal guarantees. Along these lines, the authors in \cite{SadraCDC17} leverage techniques from adaptive control to control discrete-time systems with uncertain parameters from LTL specifications; however, the approach is abstraction-based. On the other hand, the authors in \cite{SadraHSCC18} use input-output data to construct a set-valued piecewise affine model of an uncertain system that is then used in a MPC framework to synthesize a controller from an STL formula. Although this approach can provide formal guarantees for uncertain systems, identifying the uncertain system is computationally intensive and cannot be performed during run-time.

In this paper, we seek to combine optimality, correctness, and learning to synthesize approximately optimal controllers for uncertain systems with formal guarantees \emph{online}. Specifically, we consider deterministic continuous-time nonlinear control-affine systems with unknown drift terms subject to specifications expressed as scLTL formulas defined over a finite set of observations. To link the continuous behavior of the system with the discrete behavior of the scLTL formula's FSA, we form a hybrid system \cite{Teel_book} and associate with each observation of the formula a subset of the state-space, referred to as \emph{regions of interest} (ROI) throughout this paper. Leveraging the approach from \cite{CohenCDC20}, we construct a sequence of reach-avoid optimal control problems to regulate the system from ROI to ROI while avoiding ``bad" regions in the state-space. We then show that the value functions obtained from solving the corresponding sequence of HJB equations can be used to form a hybrid barrier certificate \cite{DimosACC18} to establish formal guarantees. However, solving these HJB equations can be computationally intractable and requires perfect knowledge of the system dynamics. To address this challenge, we leverage a computationally efficient model-based reinforcement learning approach \cite{RushiAutomatica16,RushiAutomatica16_StaF,RushiTNNLS19} to learn the unknown dynamics and value functions simultaneously \emph{online}. We show that, under suitable conditions, the resulting approximate optimal control policy provides formal guarantees of correctness with respect to the given specification.

The remainder of this paper is organized as follows. Section \ref{sec:prelim} covers preliminaries related to scLTL and hybrid systems \cite{Teel_book}. Section \ref{sec:formulation} formulates the problem. Sections \ref{sec:DTA} and \ref{sec:opt_ctr} present our optimal control framework. Section \ref{sec:RL} outlines our reinforcement learning-based solution approach. Section \ref{sec:sim} provides numerical examples that support the theoretical findings and Section \ref{sec:conclusion} contains concluding marks and a discussion on future research directions.

\section{Preliminaries}\label{sec:prelim}
\subsection{Notation}
For a set $\mathcal{C}$, the notations $\partial\mathcal{C},\,\text{Int}(\mathcal{C}),\,\text{cl}(\mathcal{C}),\,2^\mathcal{C}$ denote the boundary, interior, closure, and power set of $\mathcal{C}$, respectively. A continuous function $\alpha\,:\,\mathbb{R}_{\geq0}\rightarrow\mathbb{R}_{\geq0}$ is a class $\mathcal{K}$ function if $\alpha(0)=0$ and it is strictly increasing. The operator $\|(\cdot)\|$ denotes the 2-norm and $\lambda_{\min}\{(\cdot)\},\,\lambda_{\max}\{(\cdot)\}$ denote the minimum and maximum eigenvalues of a matrix, respectively. The partial derivative of a continuously differentiable function $f(x_1,x_2,\dots)$ with respect to its $i^{th}$ argument $x_i$ is denoted by $\nabla_{x_i}f(x_1,x_2,\dots)\triangleq\tfrac{\partial f}{\partial {x_i}}(x_1,x_2,\dots)$. We use $\mathcal{F}\,:\,\mathbb{R}^n\rightrightarrows\mathbb{R}^n$ to indicate that $\mathcal{F}(x)$ is a set-valued mapping.

\subsection{Syntactically co-safe Linear Temporal Logic}
In this paper, we consider control specifications in the form of \emph{syntactically co-safe linear temporal logic} (scLTL) formulas:

\begin{definition}[scLTL syntax \cite{Belta_book}]\label{def:scLTL}
Syntactically co-safe LTL formulas are recursively defined over a set of observations $O$ as
\begin{equation}\label{eq:scLTL}
    \phi=\top\,|\,o\,|\,\neg o\,|\,\phi_1\wedge\phi_2\,|\,\bigcirc\phi\,|\,\phi_1U\,\phi_2,
\end{equation}
where $o\in O$ is an observation, $\top$ is the Boolean constant True, $\neg$ and $\wedge$ stand for Boolean  negation and conjunction, $\bigcirc$ and $U$ are the ``next" and ``until" temporal operators, and $\phi$, $\phi_1$, and $\phi_2$ are scLTL formulas.
\end{definition}

Other temporal operators, such as the eventually operator $\Diamond$, can be derived using \eqref{eq:scLTL} (see \cite[Ch. 2]{Belta_book} for details). scLTL formulas can be used to specify rich behaviors for dynamical systems such as reachability and avoidance through formulas such as $\phi=(\Diamond o_1 \vee \Diamond o_2) \wedge (\neg o_3 U (o_1 \vee o_2)$, which requires a system to eventually reach $o_1$ or $o_2$ and never visit $o_3$. Traditionally, the semantics of scLTL formulas are defined over infinite words over $2^O$; however, in this paper, we adopt the approach taken in \cite{Belta_book} and consider formulas defined over infinite words over $O$. As it will become clear later, this is not restrictive, as we will use $O$ to label non-overlapping regions in the state space of a dynamical system, and at most one observation can be satisfied at any time. Although the semantics of scLTL formulas are interpreted over infinite words, their satisfaction can be decided over finite prefixes of such words. Specifically, any infinite word that satisfies an scLTL formula $\phi$ defined over $O$ also contains a finite prefix such that all infinite words that contain this prefix satisfy $\phi$. The language of $\phi$, denoted by $\mathcal{L}_{\phi}$, is the set of all finite prefixes of infinite words that satisfy $\phi$. For any scLTL formula $\phi$ over $O$, there exists a finite state automaton~\footnote{In this paper, all FSAs are assumed to be deterministic. Note that any non-deterministic FSA can be easily determinized.} $\mathcal{A}$ (see Def. \ref{def:FSA} below) that accepts $\mathcal{L}_{\phi}$ \cite{scLTL}, which can be constructed using tools such as \emph{scheck2} \cite{scheck2}. 

\begin{definition}[Finite state automaton \cite{Belta_book}]\label{def:FSA}
A finite state automaton (FSA) is a tuple $\mathcal{A}=(S,s_0,O,p,S_f)$ where $S$ is a finite set of states, $s_0\in S$ is the initial state, $O$ is the input alphabet, $p\,:\,S\times O\rightarrow S$ is a transition function, and $S_f\subseteq S$ is the set of accepting states. 
\end{definition}

The semantics of an FSA are defined over finite words over $O$. An input word $w_O=w_O(1)w_O(2),\dots,w_O(N)$, $w_O(i)\in O$, $i=1\ldots, N$ produces a run $w_S=w_S(1)w_S(2),\dots w_S(N+1)$, $w_S(i)\in S$, $i=1\ldots, N+1$, such that $w_S(1)=s_0$ and $p(w_S(i),w_O(i))=w_S(i+1)$ for all $i=1,2,\dots N$. We say that a word $w_O$ is \emph{accepted} by the FSA if and only if the run $w_S$ produced by $w_O$ ends in an accepting state ($w_S(N+1)\in S_f)$. Throughout this paper, we assume that a scLTL formula is feasible, in the sense that, in the corresponding FSA, there exists a path from the initial state $s_0$ to the set of accepting states $S_f$ \cite[Assumption 1 ]{DimosACC18}.

\subsection{Hybrid Dynamical Systems}\label{sec:hybrid}

Following \cite{Teel_book}, a hybrid system with state $x\in\mathbb{R}^n$ is a tuple $\mathcal{H}=(\mathcal{F},\mathcal{G},\mathcal{C},\mathcal{D})$, where $\mathcal{F}\,:\,\mathbb{R}^n\rightrightarrows\mathbb{R}^n$ is the flow map, $\mathcal{G}\,:\,\mathbb{R}^n\rightrightarrows\mathbb{R}^n$ is the jump map, $\mathcal{C}\subset\mathbb{R}^n$ is the flow set, and $\mathcal{D}\subset\mathbb{R}^n$ is the jump set. The semantics of $\mathcal{H}$ are given by
\begin{subequations}\label{eq:Hybrid}
\begin{empheq}[left={\mathcal{H}:\empheqlbrace}]{align}
    \dot{x}\in\mathcal{F}(x), &\quad x\in\mathcal{C}, \label{eq:flow}\\
    x^+\in\mathcal{G}(x), &\quad x\in\mathcal{D}.\label{eq:jump}
\end{empheq}
\end{subequations}
In essence, when the state $x$ belongs to the flow set $\mathcal{C}$, it continuously evolves according to the differential inclusion in \eqref{eq:flow} and when it belongs to the jump set $\mathcal{D}$ it changes, discretely, according to the difference inclusion in \eqref{eq:jump}. A \emph{hybrid time domain}, denoted by $E\subset\mathbb{R}_{\geq0}\times\mathbb{N}$, is the union of a (possibly infinite) sequence of intervals of the form $[t_j,t_{j+1}]\times\{j\}$ \cite[Def. 2.3]{Teel_book}. A \emph{hybrid arc} is a function $\xi\,:\,E\rightarrow\mathbb{R}^n$ such that for each $j\in\mathbb{N}$, $t\mapsto\xi(t,j)$ is locally absolutely continuous on the interval $\mathcal{I}^j\triangleq\{t\,:\,(t,j)\in E\}$ \cite[Def. 2.4]{Teel_book}. If a hybrid arc $\xi(t,j)$ satisfies additional conditions outlined in \cite[Def. 2.6]{Teel_book} then $\xi$ is said to be a \emph{solution} to \eqref{eq:Hybrid}. A solution to \eqref{eq:Hybrid} is said to be \emph{maximal}, denoted by $\xi\in\mathcal{S}_{\mathcal{H}}$, if it cannot be extended \cite[Def. 2.7]{Teel_book}. A complete exposition on the class of hybrid systems defined in \eqref{eq:Hybrid} is beyond the scope of this paper and we refer the unfamiliar reader to \cite[Ch. 2]{Teel_book} for an excellent introduction. Similar to \cite{DimosACC18,DimosTAC20}, in this paper, we leverage this hybrid modeling framework to combine the continuous dynamics of the plants under consideration with the discrete dynamics of the FSA corresponding to a given scLTL formula.

\begin{definition}[Eventuality \cite{DimosACC18}]\label{def:eventual}
Consider a hybrid system as in \eqref{eq:Hybrid} and a closed set $\mathcal{R}\subset\mathcal{C}\cup\mathcal{D}$. The eventuality property for $\mathcal{H}$ with respect to $\mathcal{R}$ holds if for each maximal solution $\xi\in\mathcal{S}_{\mathcal{H}}$ there exists some finite time $T\geq0$ and finite number of jumps $J\geq0$ such that:
\begin{enumerate}
    \item For $t+j<T+J$ we have $\xi(t,j)\in\mathcal{C}\cup\mathcal{D}$,
    \item For $t=T$ and $j=J$ we have $\xi(t,j)\in\mathcal{R}$.
\end{enumerate} 
\end{definition}

Intuitively, the eventuality property for a hybrid system $\mathcal{H}$ with respect to a set $\mathcal{R}$ holds if the solutions to $\mathcal{H}$ reach $\mathcal{R}$ in finite time and remain within the flow and jumps sets for all times beforehand. We will use this property in conjunction with the acceptance condition of a FSA to enforce a temporal logic specification. 
As shown in \cite{DimosACC18}, the eventuality property can be certified by finding a \emph{hybrid barrier certificate}, defined as follows:

\begin{definition}[Hybrid barrier certificate \cite{DimosACC18}]\label{def:barrier}
Let $\mathcal{R}\subset\mathcal{C}\cup\mathcal{D}$ be a closed set satisfying $\mathcal{G}(\text{cl}(\mathcal{D}\backslash\mathcal{R}))\subset\mathcal{C}\cup\mathcal{D}$. A function $V_B$ is a \emph{hybrid barrier certificate} for a hybrid system $\mathcal{H}$ with respect to $\mathcal{R}$ if it is continuous on $\text{cl}(\mathcal{C}\backslash\mathcal{R})\cup\text{cl}(\mathcal{D}\backslash\mathcal{R})$, differentiable on an open neighborhood of $\text{cl}(\mathcal{C}\backslash\mathcal{R})$, and there exists a positive constant $\epsilon_b\in\mathbb{R}_{>0}$ such that the following conditions hold:
\begin{enumerate}
    \item $V_B$ is bounded from below on $\text{cl}(\mathcal{C}\backslash\mathcal{R})\cup\text{cl}(\mathcal{D}\backslash\mathcal{R})$,
    \item $V_B$ is radially unbounded on $\text{cl}(\mathcal{C}\backslash\mathcal{R})\cup\text{cl}(\mathcal{D}\backslash\mathcal{R})$,
    \item $\nabla_x V_B(x)F(x) <-\epsilon_b$ for all $x\in\text{cl}(\mathcal{C}\backslash\mathcal{R})$ and $F\in\mathcal{F}(x)$,
    \item $V_B(G(x))-V_B(x)<-\epsilon_b$ for all $x\in\text{cl}(\mathcal{D}\backslash\mathcal{R})$ and $G\in\mathcal{G}(x)$.
\end{enumerate}
\end{definition}

The following theorem shows that the existence of a hybrid barrier certificate for $\mathcal{H}$ enforces the eventuality property:

\begin{theorem}[\cite{DimosACC18}]\label{thm:barrier}
Consider a hybrid system \eqref{eq:Hybrid} that satisfies the hybrid basic conditions\footnote{The hybrid basic conditions in \cite[Assumption 6.5]{Teel_book} are mild regularity assumptions that ensure $\mathcal{H}$ is well-posed. If $\mathcal{F}$ and $\mathcal{G}$ are single-valued mappings (i.e., differential and difference equations of the form $\dot{x}=F(x)$ and $x^+=G(x)$) then these conditions reduce to $F\,:\,\mathbb{R}^n\rightarrow\mathbb{R}^n$ and $G\,:\,\mathbb{R}^n\rightarrow\mathbb{R}^n$ being continuous on $\mathcal{C}$ and $\mathcal{D}$, respectively.} \cite[Assumption 6.5]{Teel_book} and let $\mathcal{R}\subset\mathcal{C}\cup\mathcal{D}$ be a closed set that satisfies the properties detailed in Definition \ref{def:barrier}. If there exists a hybrid barrier certificate for $\mathcal{H}$ with respect to $\mathcal{R}$, then the eventuality property holds for $\mathcal{H}$ with respect to $\mathcal{R}$.
\end{theorem}

\section{Problem Formulation}
\label{sec:formulation}
Consider a continuous-time nonlinear affine control system
\begin{equation}\label{eq:dyn}
    \dot{x}=f(x)+g(x)u,\quad x_0=x(0),
\end{equation}
where $x(t)\in\mathbb{R}^n$ is the system state, $u(t)\in\mathbb{R}^m$ is the control input, $f\,:\,\mathbb{R}^n\rightarrow\mathbb{R}^n$ models the \emph{unknown} system drift, and $g\,:\,\mathbb{R}^n\rightarrow\mathbb{R}^{n\times m}$ captures the \emph{known} control directions. The functions $f$ and $g$ are assumed to be locally Lipschitz with $f(0)=0$ and $0<\|g\|\leq\bar{g}$ where $\bar{g}\in\mathbb{R}_{>0}$. Let $O$ be a finite set of observations where each observation $o\in O$ corresponds to a \emph{region of interest} (ROI) of the form
\begin{equation}\label{eqn:ROI}
    \mathcal{R}_{o}\triangleq\{x\in\mathbb{R}^n\,:\,h_o(x)\geq0 \},
\end{equation}
for all $o\in O$, where $h_o\,:\,\mathbb{R}^n\rightarrow\mathbb{R}$ is a continuously differentiable function.

\begin{assumption}\label{assumption:ROI}
Each ROI $\mathcal{R}_o$ is non-empty and has no isolated points, i.e., $\text{Int}(\mathcal{R}_o)\neq\emptyset$ and $\text{cl}(\text{Int}(\mathcal{R}_o))=\mathcal{R}_o$ for all $o\in O$. Additionally, all ROI are disjoint in the sense that $\mathcal{R}_{o}\cap\mathcal{R}_{o'}=\emptyset,\,\forall o\neq o'$.
\end{assumption}

In this paper, we consider the following problem:
\begin{problem}\label{prob:informal}
Given a dynamical system \eqref{eq:dyn} and an scLTL formula $\phi$ \eqref{eq:scLTL} defined over a finite set of observations $O$ labeling ROIs as in \eqref{eqn:ROI}, find a control policy $u$ such that the trajectory of \eqref{eq:dyn} satisfies $\phi$.
\end{problem}

To address Problem \ref{prob:informal}, we first translate the scLTL formula $\phi$ into an FSA $\mathcal{A}$ that accepts the language satisfying $\phi$. We then compose the discrete behavior of $\mathcal{A}$ with the continuous behavior of \eqref{eq:dyn} to form a hybrid system $\mathcal{H}$ as in \eqref{eq:Hybrid}. Given this formulation, Problem \ref{prob:informal} becomes equivalent to certifying the eventuality property for $\mathcal{H}$ with respect to the set of accepting FSA states $S_f$. To synthesize a correct control policy, we construct a sequence of optimal control problems (OCPs) corresponding to the sequence of FSA states in the accepting run $w_S$ with the goal of driving the continuous trajectory of the plant through the ROIs so as to generate the accepting input word $w_O$. As a main contribution of this paper, we show that the sequence of optimal value functions obtained from solving these OCPs in combination with a distance to acceptance function defined over the FSA forms a hybrid barrier certificate to establish correctness with respect to the given scLTL formula. This formulation is related to those in \cite{TopcuCDC16,TopcuHSCC17,VamvoudakisACC20}, in the sense that we translate an scLTL formula into a corresponding FSA and then use the FSA's accepting run to construct a sequence of OCPs. The computational demands of solving these OCPs, as well as the uncertainty associated with the system drift $f$, motivates us in Sec. \ref{sec:RL} to introduce an approximate solution approach based on ADP and RL that does not require full knowledge of the system dynamics.

\section{Hybrid Systems Formulation}\label{sec:DTA}
To formalize Problem \ref{prob:informal} we turn to the setting of hybrid systems as detailed in Sec. \ref{sec:hybrid}. To this end, consider the FSA $\mathcal{A}=(S,\,s_0,\,O,\,p,\,S_f)$ corresponding to a scLTL formula $\phi$ over $O$ as a discrete-time dynamical system of the form
\begin{equation}\label{eq:p}
    s(j+1)=p(s(j),o(j)),
\end{equation}
where $j\in\mathbb{N}$, $s(j)\in S$ is the FSA state, and $o(j)\in O$  acts as an input to \eqref{eq:p}. Following the approach in \cite{DimosACC18}, we then combine the continuous-time dynamics in \eqref{eq:dyn} with the discrete-time dynamics in \eqref{eq:p} to obtain a hybrid system with state $x_{\mathcal{H}}\triangleq[x^T,\,s,\,o]^T$ of the form

\begin{subequations}\label{eq:H_nom}
    \begin{equation}
        \begin{bmatrix}
        \dot{x}(t,j)\\\dot{s}(t,j)\\\dot{o}(t,j)
        \end{bmatrix}
        =
        \begin{bmatrix}
        f(x(t,j))+g(x(t,j))u\\0\\0
        \end{bmatrix},\quad
        (x,s,o)\in\mathcal{C},
    \end{equation}
    \begin{equation}
        \begin{bmatrix}
        x(t,j+1)\\s(t,j+1)\\o(t,j+1)
        \end{bmatrix}
        =
        \begin{bmatrix}
        x(t,j)\\p(s(t,j),o(t,j))\\v(t,j)
        \end{bmatrix},\quad
        (x,s,o)\in\mathcal{D}.
    \end{equation}
\end{subequations}
Here, $x(t,j)$ represents the current state of the plant, $s(t,j)$ is the current state of the FSA, $o(t,j)$ is the next observation that should appear in the accepting word of the FSA, and $v(t,j)$ is a policy to be defined that selects the subsequent observation to be generated. Before explicitly defining $\mathcal{C},\,\mathcal{D}$ we introduce the set $O_s\subseteq O$, which is indexed by $s$ and is defined as
\begin{equation}\label{eq:Os}
    O_s\triangleq\{o\in O\,:\,p(s,o)\neq\emptyset\},\quad\forall s\in S.
\end{equation}
In essence, $O_s$ is the set of all \emph{admissible} observations for a given $s\in S$, in the sense that, when $o\in O_s$ is input to the FSA in state $s$ there exists an outgoing transition\footnote{Since all FSAs are assumed to be deterministic there exists at most one outgoing transition from $s$ for a given $o\in O_s$.} from $s$ to $p(s,o)$. Using \eqref{eq:Os} we define $\mathcal{C}$ and $\mathcal{D}$ for \eqref{eq:H_nom} as
\begin{subequations}\label{eq:CD_nom}
    \begin{equation}\label{eq:C_nom}
        \mathcal{C}\triangleq\{(x,s,o)\,:\,x\in\text{cl}(\mathbb{R}^n\backslash\mathcal{R}_o),\,s\in S,\,o\in O_s\},
    \end{equation}
    \begin{equation}\label{eq:D_nom}
        \mathcal{D}\triangleq\{(x,s,o)\,:\,x\in\mathcal{R}_o,\,s\in S,\,o\in O_s\},
    \end{equation}
\end{subequations}
which implies that for some $o\in O$ jumps can only occur if $x\in\mathcal{R}_o$; however, $x$ can flow through $\mathcal{R}_{o'}$ for some $o'\neq o$ without inducing a jump or generating observation $o'$. An exception to this is when the scLTL formula specifies that a certain observation $o'$ is \emph{not} to be generated at a given $s\in S$. In this situation we can impose active avoidance of $\mathcal{R}_{o'}$ by constructing a control barrier function \cite{AmesTAC17} for $\mathcal{R}_{o'}$
which is then incorporated into the subsequent control development outlined in Sec. \ref{sec:opt_ctr}. Given the above discussion, we now consider the more formal problem statement:
\begin{problem}\label{prob:main}
Given the hybrid system in \eqref{eq:H_nom}, find a set of control policies $(u,v)$ such that the eventuality property for \eqref{eq:H_nom} holds with respect to
\begin{equation}\label{eq:Rf}
    \mathcal{R}_f\triangleq\{(x,s,o)\,:\,x\in\mathbb{R}^n,\,s\in S_f,\,o\in O_s \}.
\end{equation}
\end{problem}

Our first step to solving Problem \ref{prob:main} is to specify the discrete-time policy $v$, which is designed to select subsequent observations that produce an accepting word of the FSA. To this end, consider the FSA $\mathcal{A}$ as a digraph where each state $s\in S$ represents a node and each observation $o\in O_s$ labels an edge from $s$ to $p(s,o)$. The policy $v$ is then selected so as to minimize the \emph{distance to acceptance} (DTA) metric \cite{BeltaAutomatica14,SerlinIROS18}, which, for some $s\in S$, represents the minimum number of transitions required to reach the set of accepting states $S_f$. Before defining the DTA, we define the distance between any two states in the FSA as \cite{SerlinIROS18}
\begin{equation}
    d(s,s')\triangleq\left\{
    \begin{aligned}
        \min_{\varphi\in\Upsilon(s,s')}\mathcal{N}(\varphi),\quad & \text{if }\Upsilon(s,s')\neq\emptyset,\\
        \infty,\quad & \text{if }\Upsilon(s,s')=\emptyset,
    \end{aligned}
    \right.
\end{equation}
where $\Upsilon$ is the set of states between $s$ and $s'$, $\varphi$ is a single path through $\Upsilon$, and $\mathcal{N}(\varphi)$ is the number of states in $\varphi$. The DTA can then be computed as 
\begin{equation}\label{eq:DTA}
    V_d(s)\triangleq\min_{s_f\in S_f}d(s,s_f).
\end{equation}
Note that $V_d$ is positive definite with respect to $S_f$ (i.e., $V_d(s)>0,\,\forall s\notin S_f$ and $V_d(s)=0,\,\forall s\in S_f$) and radially unbounded in the sense that $V_d(s)\rightarrow\infty$ as $d(s,s_f)\rightarrow\infty,\,\forall s_f\in S_f$. Using the DTA, the discrete-time control policy $v$ is selected as any element of $O_{p(s,o)}$ such that the DTA strictly decreases across jumps as \cite{DimosACC18}
\begin{multline}\label{eq:v}
    v(t,j)\in O_{p(s,o)}^d\triangleq\{v'\in O_{p(s,o)}\,:\\ \,V_d(p(p(s,o),v'))<V_d(p(s,o))\,\text{if }p(s,o)\notin S_f \},
\end{multline}
where $V_d(p(s,o))$ is the DTA resulting from transitioning from the current state $s$ to $p(s,o)$ by inputting observation $o\in O_s$ and $V_d(p(p(s,o),v'))$ is the DTA resulting from transitioning from $p(s,o)$ to $p(p(s,o),v')$ by inputting observation $v'\in O_{p(s,o)}$. In the context of this paper, the DTA can be shown to satisfy the properties of a barrier certificate for a hybrid system consisting of the purely discrete behavior of the FSA with respect to $S_f$ \cite[Lemma 3]{DimosTAC20}. Importantly, this implies that the sequence of observations selected by the policy in \eqref{eq:v} produces an accepting word $w_O$ for a given scLTL formula.

\section{Reach-Avoid Optimal Control}\label{sec:opt_ctr}
We now transition to the problem of synthesizing a continuous control strategy $u$ for the hybrid system such that the continuous trajectory of the plant $x(t,j)$ produces an accepting word for the FSA. In this section, we show how this can be accomplished by associating with each $j\in\mathbb{N}$ a suitably constructed optimal control problem. Importantly, we show that the \emph{value function} of each optimal control problem can be combined with the DTA from the previous section to construct a barrier certificate for the hybrid system in \eqref{eq:H_nom} with respect to the set in \eqref{eq:Rf}.

\subsection{System Transformation}
Consider the current state $x_{\mathcal{H}}(t,j)$ of the hybrid system in \eqref{eq:H_nom} and recall that $o(t,j)$ represents the observation that \emph{should} be generated next so as to satisfy the given scLTL formula. For ease of exposition, we drop the explicit dependence on $t,\,j$ unless otherwise needed for clarity. To ensure the state of the plant $x$ can be regulated to the ROI $\mathcal{R}_{o}$ corresponding to $o$, we perform a system transformation so that $\mathcal{R}_{o}$ contains an equilibrium point for \eqref{eq:dyn}. To this end, consider some point $x_d\in\text{Int}(\mathcal{R}_{o})$ and define the regulation error $e\triangleq x-x_d$. Additionally, define the invertible mapping $z\,:\,\mathbb{R}^n\rightarrow\mathbb{R}^n$ as $z(e)\triangleq e+x_d$, where $z^{-1}(x)=x-x_d$. To ensure the following system transformation can be performed we make the following assumption:

\begin{assumption}[\cite{RushiTNNLS16}]\label{assumption:g}
The input dynamics matrix $g(x)$ has full column rank and the function $g^{\dag}(x)\triangleq(g(x)^Tg(x))^{-1}g(x)^T$ is bounded and locally Lipschitz.
\end{assumption}

Following the formulation in \cite{RushiTNNLS16}, we decompose the control policy into two terms: a feedforward term $u_d\in\mathbb{R}^m$ that ensures $x_d$ is an equilibrium point for \eqref{eq:dyn} and a feedback term $\mu\in\mathbb{R}^m$ to regulate $x$ to $x_d$, as $u=u_d+\mu$. The steady-state control policy corresponding to $x_d$ is then $u_d=-g^{\dag}(x_d)f(x_d)$ \cite{RushiTNNLS16}, which is a constant value for each $o\in O$. With this steady-state policy, we define an auxiliary dynamical system as
\begin{equation}\label{eq:alt_dyn}
    \dot{e}=\dot{x}=\underbrace{f(z(e))+g(z(e))u_d}_{F(e)} + \underbrace{g(z(e))}_{G(e)}\mu,
\end{equation}
where $F(0)=0$ and $F(e)$ is locally Lipschitz. 

\subsection{Constructing the Cost Function}
Given the system transformation in the previous subsection, the present goal is to determine a feedback control policy $\mu\,:\,\mathbb{R}^n\rightarrow\mathbb{R}^m$ for the auxiliary system \eqref{eq:alt_dyn} such that $e\rightarrow0\implies x\rightarrow\mathcal{R}_o$. Simultaneously, for a given FSA state $s$, this policy must ensure that any negated observations $\neg o$ are \emph{not} generated (i.e., certain $\mathcal{R}_o$ must be avoided while in state $s$). To address this challenge, we associate with each state of the FSA a reach-avoid optimal control problem, where the objective is to find an optimal feedback policy $\mu^*\,:\,\mathbb{R}^n\rightarrow\mathbb{R}^m$ that minimizes the following cost functional
\begin{equation}\label{eq:J}
    V(e,\mu)\triangleq\int_{t}^{\infty}r(e(\tau),\mu(\tau))d\tau,
\end{equation}
where $r\,:\,\mathbb{R}^n\times\mathbb{R}^m\rightarrow\mathbb{R}_{\geq0}$ denotes the instantaneous positive definite cost. To facilitate this approach, let $\bar{O}_s\subset O$ denote the set of all observations that are negated at the FSA state $s$. Now, for each $o\in \bar{O}_s$ consider the set $\mathcal{X}_{o}\triangleq\{x\in\mathbb{R}^n\,:\,h_o(x)\leq0 \}$ and associate with each $\mathcal{X}_o$ a control barrier function\footnote{See \cite{AmesTAC17} for the conditions that $b_o$ must satisfy to be a valid certificate for safety and \cite{AmesL4DC20} for the assumptions required to use $b_o$ in the context of uncertain systems.} \cite{AmesTAC17} of the form $b_o(x)=-1/h_{o}(x)$, which satisfies $\inf_{x\in\mathcal{X}_o}b_o(x)\geq0$ and $\lim_{x\rightarrow\partial\mathcal{X}_o}b(x)=\infty$. To ensure the optimal control problem is well-posed\footnote{Since this optimal control problem is defined over an infinite-horizon we must have $r(0,0)=0$ to ensure a finite total cost.} we then encode each $b_o,\,\forall o\in\bar{O}_s$ within a positive semi-definite recentered barrier function  (RBF) \cite{WillisAutomatica04,PanagouTAC16} as
\begin{equation}\label{eq:B}
    B(e)\triangleq\sum_{o\in\bar{O}_s}\left(b_o(z(e))-b_o(x_d)-\nabla b_o(x_d)e \right)^2,
\end{equation}
which satisfies $B(0)=0$, $B(e)>0,\,\forall e\neq 0$, and $B(e)\rightarrow\infty$ as $z(e)\rightarrow\partial\mathcal{X}_o$ for \emph{any} $o\in\bar{O}_s$. The instantaneous cost from \eqref{eq:J} is then selected as
\begin{equation}\label{eq:r}
    r(e,\mu)=Q(e) + \mu^TR\mu + B(e),
\end{equation}
where $Q\,:\,\mathbb{R}^n\rightarrow\mathbb{R}_{\geq0}$ is a positive definite and radially unbounded function, $R\in\mathbb{R}^{m\times m}$ is a symmetric positive definite matrix, and $B(e)$ is from \eqref{eq:B}. The optimal value function (cost-to-go) $V^*\,:\,\mathbb{R}^n\rightarrow\mathbb{R}_{\geq0}$ is defined as
\begin{equation}\label{eq:value}
    V^*(e)\triangleq\inf_{\mu(\tau)\,|\,\tau\in\mathbb{R}_{\geq t}}\int_{t}^\infty r(e(\tau),\mu(\tau))d\tau,
\end{equation}
and is assumed to be continuously differentiable\footnote{Although the gradient of the value function is clearly discontinuous on the boundaries of any negated ROIs, the subsequent analysis illustrates that the system trajectory remains outside of such regions for all time, implying the value function remains continuously differentiable along the resulting system trajectory (cf. \cite{PatTRO20}).}  on an open neighborhood of $\text{cl}(\mathcal{C}\backslash\mathcal{R}_f)$. Assuming that an optimal feedback policy $\mu^*\,:\,\mathbb{R}^n\rightarrow\mathbb{R}^m$ exists, the value function is characterized by the Hamilton-Jacobi-Bellman (HJB) equation
\begin{equation}\label{eq:HJB}
    0=\min_{\mu}H(e,\nabla_e V^*,\mu)=\nabla_e V^*(e)(F(e)+G(e)\mu^*)+r(e,\mu^*),
\end{equation}
with a boundary condition of $V^*(0)=0$, where $H\,:\,\mathbb{R}^n\times\mathbb{R}^{1\times n}\times\mathbb{R}^m\rightarrow\mathbb{R}$ is the Hamiltonian. The optimal policy $\mu^*(e)$ is obtained from the Hamiltonian minimization condition $\partial H/\partial\mu^*=0$ as 
\begin{equation}\label{eq:mu*}
    \mu^*(e)=-\frac{1}{2}R^{-1}G(e)^T\nabla_e V^*(e)^T.
\end{equation}

\subsection{Hybrid Barrier Certificate}
The following proposition shows that for each $j\in\mathbb{N}$, the optimal policy in \eqref{eq:mu*} regulates $x$ to $x_d$ while avoiding any ROI that should not be visited.

\begin{proposition}\label{prop:GUAS}
If Assumptions \ref{assumption:ROI}-\ref{assumption:g} hold and there exists a positive definite and continuously differentiable function $V^*$ satisfying \eqref{eq:HJB}, then the control policy in \eqref{eq:mu*} renders the origin asymptotically stable for the auxiliary dynamical system \eqref{eq:alt_dyn}. Furthermore, the set $\text{Int}(\mathcal{X}_o)$ is forward invariant.
\end{proposition}

\begin{proof}
The proof is based on the fact that
the derivative of $V^*$ along the trajectory of \eqref{eq:alt_dyn} can be upper bounded as $\dot{V}^*\leq - \psi(\|e\|)$, where $\psi$ is a class $\mathcal{K}$ function. It follows the lines of the proof of Prop. 1 in 
\cite{CohenCDC20}, and it is omitted. \end{proof}

The following theorem illustrates that the DTA from the previous section and the optimal value function from this section can be used to construct a barrier certificate to address Problem \ref{prob:main}.
\begin{theorem}\label{thm:barrier2}
Provided Assumptions \ref{assumption:ROI}-\ref{assumption:g} and the hypothesis of Proposition \ref{prop:GUAS} hold, for some positive constant $\lambda\in\mathbb{R}_{>0}$ the function
\begin{equation}
    V_B(z^{-1}(x),s,o)\triangleq V_d(s) + \lambda V^*(z^{-1}(x)),
\end{equation}
where $V_d$ is the DTA from \eqref{eq:DTA}, is a barrier certificate for \eqref{eq:H_nom} with respect to \eqref{eq:Rf}.
\end{theorem}

\begin{proof}
Recall that a function $V_B$ is a barrier certificate if it satisfies the criteria of Definition \ref{def:barrier}. The fact that the flow map of \eqref{eq:H_nom} is continuous on the flow set along with \cite[Lemma 2]{DimosTAC20} can be used to show that \eqref{eq:H_nom} satisfies the hybrid basic conditions \cite[Assumption 6.5]{Teel_book}. The continuity of $V_B$ on $\text{cl}(\mathcal{C}\backslash\mathcal{R}_f)\cup\text{cl}(\mathcal{D}\backslash\mathcal{R}_f)$ follows from the continuity of $V_d$ and $V^*$ on $\text{cl}(\mathcal{C}\backslash\mathcal{R}_f)\cup\text{cl}(\mathcal{D}\backslash\mathcal{R}_f)$. Since we've assumed $V^*$ is continuously differentiable on an open neighborhood of $\text{cl}(\mathcal{C}\backslash\mathcal{R}_f)$, then $V_B$ is continuously differentiable on an open neighborhood of $\text{cl}(\mathcal{C}\backslash\mathcal{R}_f)$. Furthermore, $V_d$ and $V^*$ are both positive definite and radially unbounded, thus $V_B$ satisfies the first two conditions of Definition \ref{def:barrier}. For $(x,s,o)\in\text{cl}(\mathcal{C}\backslash\mathcal{R}_f)$ we have
\begin{equation}
\begin{aligned}
    \begin{bmatrix}
    \frac{\partial V_B}{\partial x} & \frac{\partial V_B}{\partial s} & \frac{\partial V_B}{\partial o}
    \end{bmatrix}
    \begin{bmatrix}
    F(z^{-1}(x)) + G(z^{-1}(x))\mu^*(z^{-1}(x)) \\ 0 \\ 0
    \end{bmatrix}
    \\
    =\lambda\nabla V^*(z^{-1}(x))\Big(F(z^{-1}(x)) + G(z^{-1}(x))\mu^*(z^{-1}(x)) \Big)
    \\
    = -\lambda\Big(Q(z^{-1}(x)) + B(z^{-1}(x)) \Big)\leq -\lambda\psi(\|z^{-1}(x)\|),
    \end{aligned}
\end{equation}
where the second equality follows by substituting in the HJB equation and the subsequent inequality from noting that $Q(z^{-1}(x)) + B(z^{-1}(x))$ is positive definite and can hence be lower bounded by a class $\mathcal{K}$ function $\psi\,:\,\mathbb{R}_{\geq0}\rightarrow\mathbb{R}_{\geq0}$ \cite[Lemma 4.3]{Khalil_book}. Since $\psi(0)=0$ is only attained at $x=x_d$ and $x_d\in\text{Int}(\mathcal{R}_o)$, then by the definition of $\mathcal{C}$ in \eqref{eq:C_nom}, there exists some positive constant $\epsilon_b\in\mathbb{R}_{>0}$ such that $\lambda\psi(\|z^{-1}(x)\|)>\epsilon_b$ for all $(x,s,o)\in\text{cl}(\mathcal{C}\backslash\mathcal{R}_f)$. Hence, the third condition of Definition \ref{def:barrier} holds. For $(x,s,o)\in\text{cl}(\mathcal{D}\backslash\mathcal{R}_f)$ we have

\small
\begin{multline}
    V_B(x(t,j+1),s(t,j+1),o(t,j+1))-V_B(x(t,j),s(t,j),o(t,j))\\=(V_d(s(t,j+1))-V_d(s(t,j))) \\+ \lambda(V^*(z^{-1}(x(t,j+1))) - V^*(z^{-1}(x(t,j)))).
\end{multline}
\normalsize
Since $V_d$ takes only integer values and strictly decreases at jumps by \eqref{eq:v} we have $V_d(s(t,j+1))-V_d(s(t,j))\leq-1$. We must now distinguish between two cases: either $V^*(z^{-1}(x(t,j+1)))$ $- V^*$ $(z^{-1}(x(t,j)))$ $<1$ or $V^*(z^{-1}(x(t,j+1)))$ $- V^*(z^{-1}(x(t,j)))\geq1$. For the former, the final condition of Definition \ref{def:barrier} is trivially satisfied by letting $\lambda=1$. For the latter, pick $\lambda$ as any sufficiently small positive constant such that 
\begin{multline}\label{eq:thm_cond4}
    V_d(s(t,j+1))-V_d(s(t,j)) \\> \lambda(V^*(z^{-1}(x(t,j+1))) - V^*(z^{-1}(x(t,j)))).
\end{multline}
By Proposition \ref{prop:GUAS}, $V^*$ is bounded for all $j\in\mathbb{N}$ and the existence of such a $\lambda$ is guaranteed. Hence the fourth condition of Definition \ref{def:barrier} holds. Since $V_B$ satisfies the conditions of Definition \ref{def:barrier}, $V_B$ is a barrier certificate and Theorem \ref{thm:barrier} can be invoked to conclude that the eventuality property for \eqref{eq:H_nom} holds with respect to \eqref{eq:Rf}.
\end{proof}

\begin{remark}
Even if $V^*(z^{-1}(x(t,j+1))) - V^*(z^{-1}(x(t,j)))\gg 1$,  $\lambda$ can always be taken to be some arbitrarily small positive constant such that the inequality in \eqref{eq:thm_cond4} holds, and hence the conditions of Definition \ref{def:barrier} hold. Note that the actual value of $\lambda$ is \emph{not} required for any part of implementing the method developed in this paper.
\end{remark}

\section{Reinforcement Learning-based Solution}\label{sec:RL}
To circumvent the computational demands of solving a sequence of OCPs and to alleviate the need to have full knowledge of the system dynamics, in this section we introduce an approximate solution approach based on ADP and RL (see \cite{ADP_survey1,ADP_survey2} for surveys of these approaches). Here, we take a model-based RL approach based on the work in \cite{RushiAutomatica16,RushiAutomatica16_StaF} wherein the value functions, optimal control policies, and system drift dynamics are all simultaneously learned online. We provide a proof of convergence of this approximation scheme and develop sufficient conditions under which formal guarantees of correctness with respect to the given specification are maintained.

\subsection{System Identification}\label{sec:sys_id}
Using the universal function approximation property of neural networks (NNs), on a compact set $\chi\subset\mathbb{R}^n$ containing the origin the unknown drift dynamics $f(x)$ can be represented as $f(x)=\theta^TY(x)+\epsilon_{\theta}(x)$ where $Y\,:\,\mathbb{R}^n\rightarrow\mathbb{R}^{p_1}$ is a vector of user-defined basis function, $\theta\in\mathbb{R}^{p_1\times n}$ are the unknown ideal NN weights, and $\epsilon_{\theta}\,:\,\mathbb{R}^n\rightarrow\mathbb{R}^n$ are the unknown function reconstruction errors\footnote{Similar to works such as \cite{RushiAutomatica16_StaF,RushiTNNLS16} it is assumed that $Y$ provides a proper basis for function approximation and hence the NN satisfies $\|\theta\|\leq\bar{\theta}<\infty$, $\sup_{x\in\chi}\|\epsilon_{\theta}(x)\|\leq\bar{\epsilon}_{\theta}$, $\sup_{x\in\chi}\|\nabla\epsilon_{\theta}(x)\|\leq\bar{\epsilon}_{\theta1}$, where $\bar{\theta},\,\bar{\epsilon}_{\theta},\,\bar{\epsilon}_{\theta1}\in\mathbb{R}_{>0}$ are positive constants. See works such as \cite{LewisAJC99} for additional details.}. Given the NN representation of $f$, the unknown NN weights are replaced with estimates $\hat{\theta}\in\mathbb{R}^{p_1\times n}$ to yield the estimated dynamics $\hat{f}(x,\hat{\theta})\triangleq\hat{\theta}^TY(x)$. To update the weights so as to minimize the approximation error $\tilde{\theta}\triangleq\theta-\hat{\theta}$, we take a concurrent learning \cite{Chowdhary,ICL} approach. In this setting, input-output data is recorded in a history stack that is used to update the weights in real-time. Provided the recorded data is rich enough, convergence of the weight estimates to a neighborhood their ideal values can be established. Specifically, we leverage an integral concurrent learning (ICL) approach \cite{ICL}, which alleviates the need to measure (or compute) state derivatives, and select the update law as \cite{PatCDC18}
\begin{subequations}\label{eq:ICL_update}
    \begin{equation}
    \begin{aligned}
        \dot{\hat{\theta}}(t)=\text{proj}\Bigg\{k_{\theta}\Gamma_{\theta}(t)\sum_{i=1}^{M}\mathcal{Y}_i\big[x(t_i)-x(t_i-\Delta t_{\theta})\\-\mathcal{U}_i-\hat{\theta}(t)^T\mathcal{Y}_i\big]^T\Bigg\}
        \end{aligned}
    \end{equation}
    \begin{equation}
        \dot{\Gamma}_{\theta}(t)=\beta_{\theta}\Gamma_{\theta}(t)-k_{\theta}\Gamma_{\theta}(t)\left(\sum_{i=1}^M\mathcal{Y}_i\mathcal{Y}_i^T\right)\Gamma_{\theta}(t),
    \end{equation}
\end{subequations}
where $k_{\theta}\in\mathbb{R}_{>0}$ is a user-defined gain, $\beta_{\theta}\in\mathbb{R}_{>0}$ is a forgetting factor, and $\text{proj}\{\cdot\}$ is a smooth projection operator that bounds the weight estimates \cite[\S 4.4]{Dixon_book}. In \eqref{eq:ICL_update} $M\in\mathbb{N}$ is the size of the history stack, $t_i\in[0,t]$ is some previous time, and the terms $\mathcal{Y}_i\triangleq\mathcal{Y}(t_i)$ and $\mathcal{U}_i\triangleq\mathcal{U}(t_i)$ are defined as $\mathcal{Y}(t_i)\triangleq\int_{\max\{0,t_i-\Delta t_{\theta} \}}^{t_i}Y(x(\tau))d\tau$ and $\mathcal{U}(t_i)\triangleq\int_{\max\{0,t_i-\Delta t_{\theta} \}}^{t_i}g(x(\tau))u(\tau)d\tau$ where $\Delta_{\theta}t\in\mathbb{R}_{>0}$ is an integration window. To ensure identification of the unknown dynamics, the history stack is required to be sufficiently rich:

\begin{assumption}[\cite{PatCDC18,ICL}]\label{assumption:ICL}
There exist constants $T_{\theta},\,\lambda_{\theta}\in\mathbb{R}_{>0}$ such that for all $t\geq T_{\theta}$, $\lambda_{\min}\big\{\sum_{i=1}^M\mathcal{Y}_i\mathcal{Y}_i^T \big\}\geq\lambda_{\theta}$.
\end{assumption}

Provided the above assumption is satisfied and $\lambda_{\min}\{ \Gamma^{-1}_{\theta}(0)\}>0$, $\Gamma_{\theta}(t)$ can be shown to satisfy $\underline{\Gamma}_{\theta}I_{p_1}\leq\Gamma_{\theta}(t)\leq\overline{\Gamma}_{\theta}I_{p_1},\,\forall t\geq0$ \cite{PatCDC18}, where $\underline{\Gamma}_{\theta},\,\overline{\Gamma}_{\theta}\in\mathbb{R}_{>0}$ and $I_{p_1}$ denotes a $p_1\times p_1$ identity matrix.  Convergence of the weight estimates to a neighborhood of their ideal values can be established using $V_{\theta}(\tilde{\theta},t)\triangleq\text{tr}\big(\tilde{\theta}^T\Gamma_{\theta}(t)\tilde{\theta}\big)$ as a Lyapunov function, which can be shown to satisfy $\dot{V}_{\theta}\leq-K_{\theta}\|\tilde{\theta}\|^2+D_{\theta}\|\tilde{\theta}\|$ where $K_{\theta}\in\mathbb{R}_{>0}$ is a constant that depends on the gains and $D_{\theta}\in\mathbb{R}_{>0}$ is a residual bound that depends on $\epsilon_{\theta}$ \cite[Thm. 1]{PatCDC18}.

\subsection{Value Function Approximation}
To estimate the value function online we leverage computationally efficient state-following kernel (StaF) functions introduced in \cite{RushiAutomatica16_StaF,RushiTNNLS19}, which provide a local representation of the value function within a small compact set that travels with the state of the system. To facilitate this approach let $\chi\subset\mathbb{R}^n$ be a compact set with $e\in\text{Int}(\chi)$ and let $B_a(e)\subset\chi$ be a closed ball of radius $a\in\mathbb{R}_{>0}$ centered at $e$. Furthermore, let $\bar{B}\,:\,\mathbb{R}^n\rightarrow\mathbb{R}_{\geq0}$ be a bounded positive semi-definite function defined as $\bar{B}(y)\triangleq\frac{B(y)}{1+B(y)}$. Then, after adding and subtracting $\bar{B}(y)$, the value function can be represented at points $y\in B_a(e)$ as \cite{RushiAutomatica16_StaF,RushiTNNLS19}
\begin{equation}
    V^*(y)=W(e)^T\sigma(y,c(e)) + \epsilon(e,y) + \bar{B}(y),
\end{equation}
where $W\,:\,\chi\rightarrow\mathbb{R}^n$ is the continuously differentiable ideal weight function, $\sigma\,:\,\chi\times \chi\rightarrow\mathbb{R}^L$ is a vector of continuously differentiable kernel functions, $c(e)\triangleq[c_1(e),\dots,c_L(e)]^T$ is a set of distinct centers, and $\epsilon\,:\chi\times\chi\rightarrow\mathbb{R}$ is the continuously differentiable function reconstruction error. Since the ideal weight function is unknown, it is replaced with estimates $\hat{W}_c,\,\hat{W}_a\,:\,\mathbb{R}_{\geq0}\rightarrow\mathbb{R}^L$ yielding the estimated value function and control policy as
\begin{subequations}
    \begin{equation}\label{eq:V_hat}
    \hat{V}(y,e,\hat{W}_c)=\hat{W}_c^T\sigma(y,c(e)) + \bar{B}(y)
\end{equation}
\begin{equation}\label{eq:mu_hat}
    \hat{\mu}(y,e,\hat{W}_a)=-\frac{1}{2}R^{-1}G(y)^T(\nabla_y\sigma(y,c(e))^T\hat{W}_a + \nabla_y\bar{B}(y)^T)
\end{equation}
\end{subequations}
The objective is then to determine values for the estimated weights $\hat{W}_c(t),\,\hat{W}_a(t)$ such that $\hat{V}$ and $\hat{\mu}$ closely approximate $V^*$ and $\mu^*$. To this end, the approximations from \eqref{eq:V_hat} and \eqref{eq:mu_hat}, as well the estimated dynamics $\hat{f}=\hat{\theta}^TY(z(e))$, are substituted into the HJB equation \eqref{eq:HJB} to yield the Bellman error (BE) $\delta$ as 
\begin{multline}\label{eq:BE}
    \delta(y,e,\hat{W}_c,\hat{W}_a,\hat{\theta})\triangleq r(y,\hat{\mu}(y,e,\hat{W}_a))\\
    + \nabla_y \hat{V}(y,e,\hat{W}_c)(\hat{F}(y,\hat{\theta})+G(y)\hat{\mu}(y,e,\hat{W}_a)),
\end{multline}
where $\hat{F}(y,\hat{\theta})\triangleq\hat{\theta}^TY(z(y))-g(z(y))g^{\dag}(x_d)\hat{\theta}^TY(x_d)$. Since the HJB equation in \eqref{eq:HJB} is equal to zero under optimal conditions, the goal of the subsequent learning strategy is to update  $\hat{W}_c(t),\,\hat{W}_a(t)$ so as to drive $\delta$ to zero.

\subsection{Weight Update Laws}
To update the weight estimates, at every time $t\in\mathbb{R}_{\geq 0}$ the BE is evaluated as $\delta_t\triangleq\delta(e,e,\hat{W}_c,\hat{W}_a,\hat{\theta})$, which implies that the control policy that influences the auxiliary dynamics in \eqref{eq:alt_dyn} is given by
\begin{equation}\label{eq:mu_t}
    \mu=\hat{\mu}(e,e,\hat{W}_a),
\end{equation}
where $\hat{\mu}$ is defined as in \eqref{eq:mu_hat}. Hence, the controller that influences the \emph{original system} in \eqref{eq:dyn} is
\begin{equation}
    u=-g^{\dag}(x_d)\hat{\theta}^TY(x) + \mu,
\end{equation}
where $\mu$ is given by \eqref{eq:mu_t}. To relax the restrictive persistence of excitation (PE) condition that is traditionally required to identify the value function, at every time $t\in\mathbb{R}_{\geq 0}$ off-policy trajectories are sampled using the family of functions $\{e_i\,:\,\mathbb{R}^n\times\mathbb{R}_{\geq0}\rightarrow\mathbb{R}^n\}_{i=1}^{N}$, where each function $e_i(e(t),t)\in B_a(e(t))$ defines a mapping from the current state $e$ to an unexplored point in $B_a(e)$. The BE is then evaluated at each sampled point as $\delta_{ti}\triangleq\delta(e_i,e,\hat{W}_c,\hat{W}_a,\hat{\theta})$. Using both $\delta_t$ and $\delta_{ti}$ the value function weights are then updated using a recursive least-squares update law with a forgetting factor as
\begin{equation}\label{eq:Wc_dot}
    \dot{\hat{W}}_c=-k_{c1}\Gamma\frac{\omega}{\rho^2}\delta_t-\frac{k_{c2}}{N}\Gamma\sum_{i=1}^N\frac{\omega_i}{\rho_i^2}\delta_{ti},
\end{equation}
where $\Gamma(t)\in\mathbb{R}^{L\times L}$ is a dynamic gain matrix, $\rho(t)\triangleq1+\gamma_1\omega(t)^T\omega(t)$ is a normalization term, $k_{c1},\,k_{c2},\,\gamma_1\in\mathbb{R}_{>0}$ are positive constant gains and $\omega(t)$ is a regressor vector defined as\footnote{All terms with the subscript $i$ denote that a function is evaluated at the sampled state $e_i$. For example, $\omega_i\triangleq\nabla\sigma(e_i,c(e))\Big(\hat{F}(e_i,\hat{\theta}) + G(e_i)\hat{\mu}(e_i,e,\hat{W}_a)\Big)$.} 
\begin{equation}
    \omega(t)\triangleq\nabla_e\sigma(e(t),c(e(t)))\Big(\hat{F}(e(t),\hat{\theta}(t)) + G(e(t))\mu(t)\Big).
\end{equation}
The matrix $\Gamma(t)$ in \eqref{eq:Wc_dot} evolves according to
\begin{equation}\label{eq:gamma_dot}
    \dot{\Gamma}=\beta\Gamma-k_{c1}\Gamma\Lambda\Gamma - \frac{k_{c2}}{N}\Gamma\sum_{i=1}^N\Lambda_i\Gamma,
\end{equation}
where $\beta\in\mathbb{R}_{>0}$ is a forgetting factor and $\Lambda(t)\triangleq\frac{\omega(t)\omega(t)^T}{\rho^2(t)}$. Based on the analysis in the subsequent subsection the control policy weights are selected to evolve according to
\begin{multline}\label{eq:Wa_dot}
    \dot{\hat{W}}_a=-k_{a1}(\hat{W}_a-\hat{W}_c)-k_{a2}\hat{W}_a\\
    + \frac{k_{c1}G_{\sigma}^T\hat{W}_a\omega^T\hat{W}_c}{4\rho^2} + \sum_{i=1}^N\frac{k_{c2}G_{\sigma i}^T\hat{W}_a\omega_i^T\hat{W}_c}{4N\rho_i^2},
\end{multline}
where $k_{a1},\,k_{a2}\in\mathbb{R}_{>0}$ are positive constant learning gains and $G_{\sigma}\triangleq\nabla_e\sigma(e,c(e))G_R(e)\nabla_{e}\sigma(e,c(e))^T$ and $G_R\triangleq G(e)R^{-1}G(e)^T$.

\subsection{Convergence Results}
The following assumption outlines excitation conditions under which convergence of the approximation scheme to a neighborhood of the optimal solution can be established:

\begin{assumption}[\cite{RushiAutomatica16_StaF}]\label{assumption:FE}
There exist constants $T\in\mathbb{R}_{>0}$ and $\underline{c}_1,\,\underline{c}_2,\,\underline{c}_3\in\mathbb{R}_{\geq0}$ with at least one of $\underline{c}_1,\,\underline{c}_2,\,\underline{c}_3$ strictly positive such that 
\begin{subequations}
    \begin{equation}\label{eq:FE1}
        \underline{c}_1I_L\leq\frac{1}{N}\int_{t}^{t+T}\left(\sum_{i=1}^N\Lambda_i(t)\right)d\tau,\quad\forall t\in\mathbb{R}_{\geq 0}
    \end{equation}
    \begin{equation}\label{eq:FE2}
        \underline{c}_2I_L\leq\inf_{t\in\mathbb{R}_{\geq 0}}\left(\frac{1}{N}\sum_{i=1}^N\Lambda_i(t) \right)
    \end{equation}
    \begin{equation}\label{eq:FE3}
        \underline{c}_3I_L\leq\int_{t}^{t+T}\left(\sum_{i=1}^N\Lambda(t)\right)d\tau,\quad\forall t\in\mathbb{R}_{\geq 0}.
    \end{equation}
\end{subequations}
\end{assumption}
The condition in \eqref{eq:FE3} asks for the regressor $\omega$, which depends on the trajectory of the system, to be PE which is a strong condition that is notoriously difficult to verify online. On the other hand, the conditions in \eqref{eq:FE1} and \eqref{eq:FE2} are placed on $\omega_i$, which depends on the sampled trajectories and can be verified online \cite{PatTRO20}. Hence, one can select the family of functions $\{e_i\,:\,\mathbb{R}^n\times\mathbb{R}_{\geq0}\rightarrow\mathbb{R}^n\}_{i=1}^{N}$ so that $\omega_i$ satisfies \eqref{eq:FE1} or \eqref{eq:FE2} by either sampling highly oscillatory extrapolated trajectories, for \eqref{eq:FE1}, or by sampling a large number of extrapolated trajectories, for \eqref{eq:FE2}. Indeed, multiple results \cite{RushiAutomatica16_StaF,PatTRO20,PatAutomatica20} have demonstrated that Assumption \ref{assumption:FE} can be satisfied with as few as one time-varying extrapolation point. Moreover, if Assumption \ref{assumption:FE} is satisfied and $\lambda_{\min}\{\Gamma^{-1}(0)\}>0$, the update law in \eqref{eq:gamma_dot} can be used to show that the matrix $\Gamma$ satisfies $\underline{\Gamma}I_L\leq\Gamma(t)\leq\overline{\Gamma}I_L$ for all $t\geq 0$, where $\underline{\Gamma},\,\overline{\Gamma}\in\mathbb{R}_{>0}$ \cite[Lemma 1]{RushiAutomatica16_StaF}.

Now consider the following Lyapunov function candidate
\begin{equation}\label{eq:V_L}
    V_L(Z_L,t)=V^*(e)+\frac{1}{2}\tilde{W}_c^T\Gamma^{-1}(t)\tilde{W}_c + \frac{1}{2}\tilde{W}_a^T\tilde{W}_a + V_{\theta}(\tilde{\theta},t),
\end{equation}
where $Z_L\triangleq[e^T,\,\tilde{W}_c^T,\,\tilde{W}_a^T,\,\text{vec}(\tilde{\theta})^T]^T$ is a composite state vector, $\tilde{W}_c\triangleq W-\hat{W}_c$, $\tilde{W}_a\triangleq W-\hat{W}_a$ denote the weight estimation errors, and $V_{\theta}$ is from Sec. \ref{sec:sys_id}. Since $V^*$ is positive definite, $V_L$ is positive definite and can hence be bounded as $\underline{v}_l(\|Z_L\|)\leq V_L(Z_L,t)\leq\overline{v}_l(\|Z_L\|)$ where $\underline{v}_l,\,\overline{v}_l\,:\,\mathbb{R}_{\geq0}\rightarrow\mathbb{R}_{\geq0}$ are class $\mathcal{K}$ functions \cite[Lemma 4.3]{Khalil_book}. Moreover, let $\mathcal{B}_{\zeta}\subset\chi\times\mathbb{R}^{2L\times np_1}$ be a closed ball of radius $\zeta\in\mathbb{R}_{>0}$ centered at the origin and define the operator $\overline{\| (\cdot)\|}\triangleq\sup_{\pi\in\mathcal{B}_{\zeta}}\|(\cdot)\|$. The following proposition illustrates that the developed learning strategy guarantees uniformly ultimately bounded convergence to the optimal solution.

\begin{proposition}\label{prop:UUB}
Provided Assumptions \ref{assumption:g}-\ref{assumption:FE} hold and the following sufficient conditions are satisfied
\begin{equation}\label{eq:gains}
\begin{aligned}
    \underline{c}>\frac{\varphi_{c\theta} + \varphi_{ac}}{k_{c2}}, & \quad & k_{a1}+k_{a2}>2\varphi_a+\varphi_{ac}, \\ K_{\theta}>\varphi_{c\theta}, & \quad & v_l^{-1}(\iota)<\overline{v}_l^{-1}\big(\underline{v}_l(\zeta) \big),
\end{aligned}
\end{equation}
where $\underline{c}\triangleq\frac{\underline{c}_2}{2}+\frac{\beta}{2\overline{\Gamma}k_{c2}}$, and $\varphi_{c\theta},\,\varphi_{ac},\,\varphi_a,\,\iota\in\mathbb{R}_{>0}$ are positive constants and $v_l\,:\,\mathbb{R}_{\geq0}\rightarrow\mathbb{R}_{\geq0}$ is a class $\mathcal{K}$ function developed in the subsequent proof,
then the control policy in \eqref{eq:mu_t} and the update laws in \eqref{eq:Wc_dot}, \eqref{eq:gamma_dot}, \eqref{eq:Wa_dot} guarantee that $Z_L(t)$ is uniformly ultimately bounded in the sense that $\limsup_{t\rightarrow\infty}\|Z_L(t)\|\leq\underline{v}_l^{-1}\Big(\overline{v}_l\big(v_l^{-1}(\iota) \big) \Big).$ Moreover, the set $\text{Int}(\mathcal{X}_o)$ is forward invariant.
\end{proposition}

By the definition of uniform ultimate boundedness \cite[Def. 4.6]{Khalil_book}, the above proposition implies that $x$ reaches some neighborhood of $x_d$ in \emph{finite time}, the size of which depends on the parameters associated with the learning scheme (see Remark \ref{remark:UUB}).

\begin{proof}
The derivative of $V_L$ along the trajectory of the auxiliary system in \eqref{eq:alt_dyn} is
\begin{multline}
    \dot{V}_L=\dot{V}^* + \tilde{W}_c^T\Gamma^{-1}(\dot{W}-\dot{\hat{W}}_c) - \tfrac{1}{2}\tilde{W}_c^T\Gamma^{-1}\dot{\Gamma}\Gamma^{-1}\tilde{W}_c\\
    +\tilde{W}_a^T(\dot{W}-\dot{\hat{W}}_a) - K_{\theta}\|\tilde{\theta}\|^2 + D_{\theta}\|\tilde{\theta}\|. 
\end{multline}
Using the definition of the weight estimation errors, the BE can be expressed as $\delta_t=-\omega^T\tilde{W}_c+\tfrac{1}{4}\tilde{W}_a^TG_{\sigma}\tilde{W}_a-\Delta_F\tilde{F}+\Delta$ where $\Delta_F\triangleq W^T\nabla_e\sigma+\nabla_e\bar{B}$, $\tilde{F}\triangleq\hat{F}(e,\tilde{\theta})$, and $\Delta$ consists of terms that are uniformly bounded over $\chi$. Using $\dot{W}=\nabla_e W(F+G\mu)$, the HJB \eqref{eq:HJB}, substituting in the update laws \eqref{eq:Wc_dot}, \eqref{eq:gamma_dot}, \eqref{eq:Wa_dot}, and upper bounding yields
\begin{equation}\label{eq:VL_dot1}
    \begin{aligned}
        \dot{V}_L & \leq -\psi(\|e\|) - \underline{c}k_{c2}\|\tilde{W}_c\|^2 + \big(\varphi_a-k_{a1}-k_{a2}\big)\|\tilde{W}_a\|^2 \\
        & -K_{\theta}\|\tilde{\theta}\|^2 + \varphi_{c\theta}\|\tilde{W}_c\| \|\tilde{\theta}\| + \varphi_{ac}\|\tilde{W}_a\| \|\tilde{W}_c\| + \iota_c\|\tilde{W}_c\| \\
        & + \iota_a\|\tilde{W}_a\| + D_{\theta}\|\tilde{\theta}\| + \tfrac{1}{2}\overline{\|\nabla_e V^* \| \|G_R\|\|\nabla_e\sigma^TW + \nabla_e\epsilon^T\|},
    \end{aligned}
\end{equation}
where $\varphi_a,\,\varphi_{c\theta},\,\varphi_{ac},\,\iota_c,\,\iota_a\in\mathbb{R}_{>0}$ are positive constants defined in the appendix. Provided the gains are selected according to \eqref{eq:gains}, applying Young's Inequality, segregating terms, and completing the square allows \eqref{eq:VL_dot1} to be further bounded as
\begin{multline}\label{eq:VL_dot2}
    \dot{V}_L\leq -\psi(\|e\|) - \tfrac{\underline{c}k_{c2}}{2}\|\tilde{W}_c\|^2 - \tfrac{k_{a1}+k_{a2}}{2}\|\tilde{W}_a\|^2 - \tfrac{K_{\theta}}{2}\|\tilde{\theta}\|^2\\
    + \tfrac{\iota_c^2}{2\underline{c}k_{c2}} + \tfrac{\iota_a^2}{2(k_{a1}+k_{a2})} + \tfrac{D_{\theta}^2}{2K_{\theta}} + \tfrac{\overline{\|\nabla_e V^* \| \|G_R\|\|\nabla_e\sigma^TW + \nabla_e\epsilon^T\|}}{2},
\end{multline}
which implies that $V_L$ is bounded and decreasing outside of a small compact set centered at the origin such that
\begin{equation}\label{eq:VL_dot3}
    \dot{V}_L\leq-v_l(\|Z_L\|),\quad \forall \|Z_L\|>v_l^{-1}(\iota),
\end{equation}
where $v_l\,:\,\mathbb{R}_{\geq0}\rightarrow\mathbb{R}_{\geq0}$ is a class $\mathcal{K}$ function that satisfies 
\begin{equation}\label{eq:v_l}
    v_l(\|Z_L\|)\leq\tfrac{\psi(\|e\|)}{2} + \tfrac{\underline{c}k_{c2}}{4}\|\tilde{W}_c\|^2 + \tfrac{k_{a1}+k_{a2}}{4}\|\tilde{W}_a\|^2 + \tfrac{K_{\theta}}{4}\|\tilde{\theta}\|^2,
\end{equation}
and $\iota\in\mathbb{R}_{>0}$ is a positive constant defined as
\begin{equation}\label{eq:iota}
    \iota\triangleq\tfrac{\iota_c^2}{\underline{c}k_{c2}} + \tfrac{\iota_a^2}{k_{a1}+k_{a2}} + \tfrac{D_{\theta}^2}{K_{\theta}} + \overline{\|\nabla_e V^* \| \|G_R\|\|\nabla_e\sigma^TW + \nabla_e\epsilon^T\|}.
\end{equation}
The bound in \eqref{eq:VL_dot3} and the fact that $V_L$ is bounded as $\underline{v}_l(\|Z_L\|)\leq V_L(Z_L,t)\leq\overline{v}_l(\|Z_L\|)$ can be used with \cite[Thm. 4.18]{Khalil_book} to conclude that $Z_L$ is uniformly ultimately bounded. Similar to Prop. \ref{prop:GUAS}, the boundedness of $V_L$ implies that system trajectories cannot reach any ROI that are negated at the current FSA state. Since $Z_L\in\mathcal{L}_{\infty}$, then $e,\tilde{W}_c,\tilde{W}_a,\tilde{\theta}\in\mathcal{L}_{\infty}$. Since $W$ is a continuous function of $e$, then   $W\in\mathcal{L}_{\infty}$, which implies that $\hat{W}_c,\hat{W}_a\in\mathcal{L}_{\infty}$ and thus $\mu\in\mathcal{L}_{
\infty}$.
\end{proof}

\begin{remark}\label{remark:UUB}
The sufficient conditions from \eqref{eq:gains} can be satisfied by increasing the number of basis functions for value function approximation $L$, which reduces the residual error $\epsilon$, and by sampling off-policy trajectories that result in large values of $\underline{c}$. Similarly, increasing the number of basis function used to estimate the unknown dynamics $p_1$ and size of the history stack $M$ will help satisfy \eqref{eq:gains}. Additionally, selecting $R$ with a large minimum eigenvalue and increasing $k_{a2}$ and $\gamma_1$ can help satisfy \eqref{eq:gains}.
\end{remark}

Now consider an \emph{augmented} hybrid system $\hat{\mathcal{H}}=(\hat{\mathcal{F}},\hat{\mathcal{G}},\hat{\mathcal{C}},\hat{\mathcal{D}})$ with state $x_{\hat{\mathcal{H}}}\triangleq[x^T,\,\tilde{W}_c^T,\,\tilde{W}_a^T,\,\text{vec}(\tilde{\theta})^T,\,s,\,o]^T$ that captures the dynamics of the plant, automaton, and learning strategy, where
\begin{subequations}\label{eq:H_aug}
\begin{equation}\label{eq:H_augF}
\hat{\mathcal{F}}(x_{\hat{\mathcal{H}}})\triangleq
\begin{bmatrix}
\dot{x}^T & \dot{\tilde{W}}_c^T & \dot{\tilde{W}}_a^T & \text{vec}(\dot{\tilde{\theta}})^T & 0 & 0
\end{bmatrix}^T,
\end{equation}
\begin{equation}\label{eq:H_augG}
    \hat{\mathcal{G}}(x_{\hat{\mathcal{H}}})\triangleq
\begin{bmatrix}
x^T & (\tilde{W}_c^+)^T & (\tilde{W}_a^+)^T & \text{vec}(\tilde{\theta})^T & p(s,o) & v
\end{bmatrix}^T,
\end{equation}
\begin{multline}\label{eq:H_augC}
    \hat{\mathcal{C}}\triangleq\{(x,\tilde{W}_c,\tilde{W}_a,\tilde{\theta},s,o )\,:\,x\in\text{cl}(\mathbb{R}^n\backslash\mathcal{R}_o),\,\tilde{W}_c\in\mathbb{R}^L\\\tilde{W}_a\in\mathbb{R}^L,\,\tilde{\theta}\in\mathbb{R}^{p_1\times n},\,s\in S,\,o\in O_s \},
\end{multline}
\begin{multline}\label{eq:H_augD}
    \hat{\mathcal{D}}\triangleq\{(x,\tilde{W}_c,\tilde{W}_a,\tilde{\theta},s,o )\,:\,x\in\mathcal{R}_o,\,\tilde{W}_c\in\mathbb{R}^L\\\tilde{W}_a\in\mathbb{R}^L,\,\tilde{\theta}\in\mathbb{R}^{p_1\times n},\,s\in S,\,o\in O_s \}.
\end{multline}
\end{subequations}
In \eqref{eq:H_augG}, the terms $\tilde{W}_c^+,\,\tilde{W}_a^+$ denote the jump in the weight estimation errors when switching from one OCP to the next as the ideal weights corresponding to the subsequent optimal value function and control policy may be different from the previous OCP. The following theorem shows that if the ultimate bound in Proposition \ref{prop:UUB} is sufficiently small, then the Lyapunov function $V_L$ and the DTA $V_d$ from Sec. \ref{sec:DTA} can be used construct a barrier certificate for \eqref{eq:H_aug}.

\begin{theorem}\label{thm:barrier3}
Suppose Assumptions \ref{assumption:ROI}-\ref{assumption:FE} and the hypothesis of Proposition \ref{prop:UUB} hold and the positive constant $\iota$ from \eqref{eq:iota} is sufficiently small. Then, for some $\lambda\in\mathbb{R}_{>0}$ the function
\begin{equation}
    \hat{V}_B\triangleq V_d(s) + \lambda V_L(Z_L,t),
\end{equation}
where $V_d$ is from \eqref{eq:DTA} and $V_L$ is from \eqref{eq:V_L}, is a barrier certificate for \eqref{eq:H_aug} with respect to 
\begin{multline}\label{eq:Rf_hat}
    \hat{\mathcal{R}}_f\triangleq\{(x,\tilde{W}_c,\tilde{W}_a,\tilde{\theta},s,o )\,:\,x\in\mathbb{R}^n,\,\tilde{W}_c\in\mathbb{R}^L\\\tilde{W}_a\in\mathbb{R}^L,\,\tilde{\theta}\in\mathbb{R}^{p_1\times n},\,s\in S_f,\,o\in O_s \}.
\end{multline}
\end{theorem}
\begin{proof}[Proof sketch]
The proof follows the same steps as that of Theorem \ref{thm:barrier2}; however, more attention must be given to the third condition of Definition \ref{def:barrier}. Following the same steps as in Theorem \ref{thm:barrier2}, for $(x,\,\tilde{W}_c,\,\tilde{W}_a,\,\tilde{\theta},\,s,\,o)\in\text{cl}(\hat{\mathcal{C}}\backslash\hat{\mathcal{R}}_f)$ we have $\nabla{\hat{V}_B}\hat{\mathcal{F}}\leq-\lambda v_l(\|Z_L \|)$ as long as $\|Z_L\|>v_l^{-1}(\iota)$, where $v_l$ is the class $\mathcal{K}$ function and $\iota$ is the constant defined in \eqref{eq:v_l} and \eqref{eq:iota}, respectively. Hence, the third condition of Definition \ref{def:barrier} holds provided the bound $v_l^{-1}(\iota)$ is sufficiently small, which can be achieved by appropriately selecting the parameters associated with the learning scheme (see Remark \ref{remark:UUB2}). Verifying the first, second, and fourth condition of Definition \ref{def:barrier} follows closely to that of Theorem \ref{thm:barrier2} and is omitted. Given that $\hat{V}_B$ can be shown to satisfy the properties of a barrier certificate outlined in Definition \ref{def:barrier}, Theorem \ref{thm:barrier} can be invoked to show that the eventuality property holds for $\mathcal{\hat{H}}$ with respect to $\mathcal{\hat{R}}$. 
\end{proof}

\begin{remark}\label{remark:UUB2}
The condition that $\iota$ be sufficiently small is somewhat restrictive as it depends on terms such as the residual value function approximation error $\epsilon$, which are unknown and infeasible to compute in general \cite[Footnote 6]{RushiAutomatica16_StaF}. Despite the fact that it may be difficult to obtain a reasonable estimate of $\iota$, the bound $v_l^{-1}(\iota)$ is a class $\mathcal{K}$ function of $\iota$ and hence decreases with decreasing $\iota$. Based on the definition of $\iota$ in \eqref{eq:iota}, the bound can be decreased by selecting the parameters for the learning scheme as outlined in Remark \ref{remark:UUB}. Similar to previous works that have leveraged the model-based RL framework used in this paper \cite{RushiAutomatica16_StaF,PatTRO20,PatTNNLS18,PatAutomatica20,CohenCDC20}, we'll demonstrate numerically in the following section that a sufficient approximation can generally be achieved using as few as $L=n+1$ basis functions for value function approximation.
\end{remark}
\section{Numerical Examples}\label{sec:sim}

We consider a system as in \cite[Ch. 5.2]{Rushi_book}, \cite[\S VI-C]{PatTRO20} with $x=[x_1,\,x_2]^T$ and
\begin{equation}
    f(x)=\begin{bmatrix}
    -x_1 + x_2 \\ 
    -\tfrac{1}{2}x_1 - \tfrac{1}{2}x_2(1- (\cos(2x_1)+2)^2)
    \end{bmatrix}
\end{equation}
\begin{equation}
    g(x)=\begin{bmatrix}
    \sin(2x_1)+2 & 0\\
    0 & \cos(2x_1)+2
    \end{bmatrix}.
\end{equation}
To define a specification for the system, we consider non-overlapping ROI of the form $\mathcal{R}_{o_i}\triangleq\{x\in\mathbb{R}^2\,:\,\|x-x_{d,i}\|\leq r_i \}$, where $x_{d,i}\in\mathbb{R}^2$ and $r_i\in\mathbb{R}_{>0}$ denote the center and radius of the $i$th ROI, respectively, and $h_{o_i}(x)=r_i-\|x-x_{d,i}\|$, where $i\in\{1,\dots,5\}$. Using these ROI we construct a specification as an scLTL formula of the form \cite[\S 5.2]{TopcuHSCC17}
\begin{equation}\label{eq:scLTL1}
    \phi=\Diamond(\Diamond(o_1 \wedge \Diamond o_2) \vee \Diamond(o_2 \wedge \Diamond o_1) \wedge o_3 ) \wedge ((\neg o_4 \wedge \neg o_5)\, U\, o_3)
\end{equation}
which states ``visit $\mathcal{R}_{o_1}$ \textbf{then} $\mathcal{R}_{o_2}$ \textbf{or} visit $\mathcal{R}_{o_2}$ \textbf{then} $\mathcal{R}_{o_1}$ \textbf{and eventually} visit $\mathcal{R}_{o_3}$ \textbf{and don't visit} $\mathcal{R}_{o_4}$ and $\mathcal{R}_{o_5}$ \textbf{until} visiting $\mathcal{R}_{o_3}$." The FSA corresponding to the scLTL formula in \eqref{eq:scLTL1} is provided in Fig. \ref{fig:FSA}. Using the DTA from Sec. \ref{sec:DTA} we compute $w_O=o_1o_2o_3$ and $w_O=o_2o_1o_3$ as accepting words, which can be used to construct a sequence of optimal control problems as outlined in Section \ref{sec:opt_ctr}. 

\begin{figure}
    \centering
    \begin{tikzpicture}[shorten >=1pt,node distance=2cm,on grid,auto]
  \tikzstyle{every state}=[fill={rgb:black,1;white,10}]

  \node[state,initial]            (s_0)                 {$s_0$};
  \node[state]                    (s_1) [above right of=s_0]  {$s_1$};
  \node[state]                    (s_2) [below right of=s_0]  {$s_2$};
  \node[state]                    (s_3) [below right of=s_1]  {$s_3$};
  \node[state,accepting]          (s_4) [right of=s_3]        {$s_4$};
  \path[->]
  (s_0) edge                node {$o_1$}  (s_1)
  (s_0) edge                node {$o_2$}  (s_2)
  (s_0) edge  [loop below]  node {$\neg o_4 \wedge \neg o_5$}  ()
  (s_1) edge  [loop above]  node {$\neg o_4 \wedge \neg o_5$}  ()
  (s_1) edge                node {$o_2$}  (s_3)
  (s_2) edge  [loop below]  node {$\neg o_4 \wedge \neg o_5$}  ()
  (s_2) edge                node {$o_1$}  (s_3)
  (s_3) edge  [loop below]  node {$\neg o_4 \wedge \neg o_5$}  ()
  (s_3) edge                node {$o_3$}  (s_4);
\end{tikzpicture}
    \caption{FSA corresponding to the scLTL formula in \eqref{eq:scLTL1}. The double circle denotes the final state. A transition labeled by a Boolean logic formula stands for a set of transitions, each labeled by a singleton, e.g., the self transition labeled by $\neg o_4 \wedge \neg o_5$ at $s_1$ represents 3 self transitions labeled by $o_1$, $o_2$, and $o_3$, respectively.}
    \label{fig:FSA}
\end{figure}
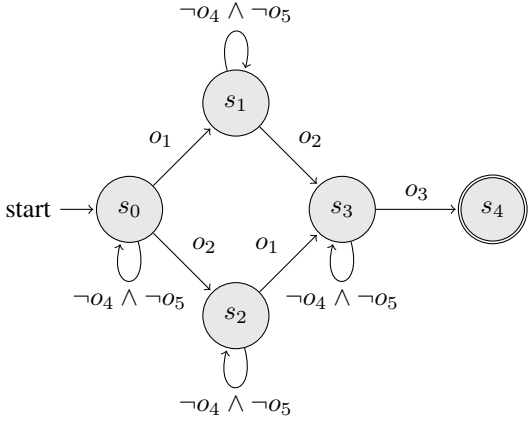

For each optimal control problem we select the cost function as in \eqref{eq:r} with $Q=\|e\|^2$ and $R=2I_2$. To ensure the trajectory never reaches $\mathcal{R}_{o_4}$ or $\mathcal{R}_{o_5}$ we construct a RBF as in \eqref{eq:B} with $b_{o_4}(x)=\tfrac{-1}{h_{o_4}(x)}$, $b_{o_5}(x)=\tfrac{-1}{h_{o_5}(x)}$ and then scale the RBF as $0.01\times B(x)$. To perform system identification we parameterize the drift dynamics as $f(x)=\theta^TY(x)+\epsilon_{\theta}(x)$ where the vector of basis functions is selected as $Y(x)=[x_1,\,x_2,\,x_2(1- (\cos(2x_1)+2)^2)]^T$ and thus the ideal weights are 
\begin{equation*}
    \theta=\begin{bmatrix}
    \theta_1 & \theta_4 \\ \theta_2 & \theta_5 \\ \theta_3 & \theta_6 
    \end{bmatrix},\,\theta_1= -1,\,\theta_2 = 1,\, \theta_3,\theta_5=0,\,\theta_4,\theta_6=-\tfrac{1}{2}.
\end{equation*}
Since an exact basis is used to represent $f(x)$, $\epsilon_{\theta}(x)=0$ and the weight estimates can be compared to their true values. The unknown dynamics are learned by using a pre-populated history stack, which is then used to estimate the unknown weights online using the update laws in \eqref{eq:ICL_update} with $k_{\theta}=15,\,\beta_{\theta}=10$, and $\Gamma_{\theta}(0)=20I_3$ where the initial weight estimates $\hat{\theta}(0)$ are drawn randomly from a uniform distribution between -5 and 5. The parameters for value function approximation are selected to be the same as in \cite[\S VII-A]{PatTNNLS18},  where $\Gamma(0)=15I_3$ and the initial weights for the value function $\hat{W}_c(0)$ and control policy $\hat{W}_a(0)$ are chosen as $\hat{W}_c(0)=\hat{W}_a(0)=4_{3\times1}$. The closed-loop system with an initial condition of $x(0)=[-2,2]^T$ is simulated for both accepting words in MATLAB/Simulink R2020b for 5 seconds running at 1000Hz on a 2019 MacBook Pro with a 1.4 GHz Quad-Core Intel i5 processor and 8 GB of memory. The computation time for a typical run is approximately 1.7 seconds as measured by MATLAB's \texttt{timeit} function.

\paragraph*{Results}
The results of the simulation are illustrated in Fig. \ref{fig:sim_x}- \ref{fig:sim_other}. In each figure the solid lines correspond to the simulation for $w_O=o_1o_2o_3$ whereas the simulation corresponding to $w_O=o_2o_1o_3$ is represented by dashed lines. In Fig. \ref{fig:sim_x} the yellow, purple, and green disks denote $\mathcal{R}_{o_1},\,\mathcal{R}_{o_2},\,\mathcal{R}_{o_3}$, respectively, and the red disks denote $\mathcal{R}_{o_4},\,\mathcal{R}_{o_5}$ where the initial condition is represented as a black square and the terminal condition is represented as a black star. The trajectory for both runs can be seen to successfully satisfy the scLTL specification by visiting the ROI in the correct order without ever entering $\mathcal{R}_{o_4},\,\mathcal{R}_{o_5}$. The evolution of the FSA state over time is provided in Fig. \ref{fig:sim_other} (top right), where the accepting words $w_O=o_1o_2o_3$ and $w_O=o_2o_1o_3$ are produced in around 1 second for each simulation. Fig. \ref{fig:sim_other} (bottom right) illustrates the evolution of the estimated value function and control policy weights. Because the optimal value function and policy are unknown the estimated weights cannot be compared to their true values. Moreover, this implies that the explicit evolution of the barrier certificate $\hat{V}_B$ is unknown since $V_L$ contains unknown terms. The evolution of the system identification error is shown in Fig. \ref{fig:sim_other} (bottom left), where the error is shown to rapidly decay to zero. Note that the history stack was populated with pre-recorded data; however, the weights are updated online during run time. The use of a pre-recorded history stack is made for ease of implementation and to improve computational efficiency. As previously mentioned, this data can be recorded online using algorithms such as those mentioned in \cite{ChowdharyACC11} at the cost of additional computation. 

\begin{figure}
    \centering
    \includegraphics[width=0.47\textwidth]{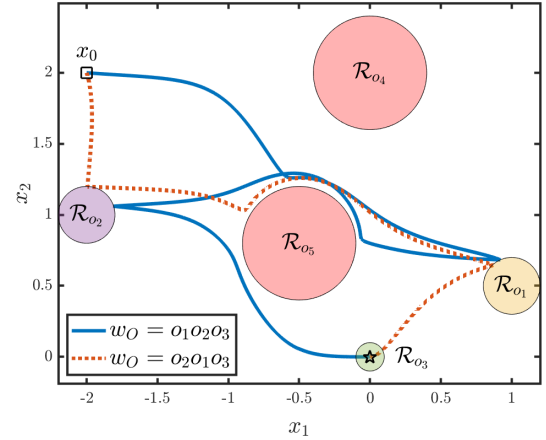}
    \caption{Phase portrait of the resulting system trajectory. The ROI $\mathcal{R}_{o_1},\,\mathcal{R}_{o_2},\,\mathcal{R}_{o_3}$ are represented by the yellow, purple, and green disks, respectively, and the red disks portray $\mathcal{R}_{o_4},\,\mathcal{R}_{o_5}$. The solid blue line represents the simulation corresponding to $w_O=o_1o_2o_3$ whereas the dashed orange line represents the simulation corresponding to $w_O=o_2o_1o_3$.}
    \label{fig:sim_x}
\end{figure}

\begin{figure}
    \centering
    \includegraphics[width=0.47\textwidth]{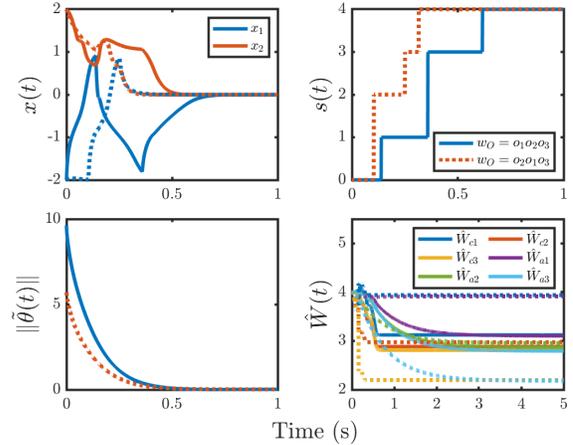}
    \caption{Evolution of the system states (top left), FSA state (top right), system ID weights (bottom left), and StaF weights (bottom right) over time.}
    \label{fig:sim_other}
\end{figure}




\section{Conclusion and Future Work}\label{sec:conclusion}
We presented a hybrid systems approach to synthesizing optimal control policies for uncertain systems from scLTL formulas. By adopting tools from hybrid systems and optimal control theory, we provided Lyapunov-like conditions for scLTL formula satisfaction and showed that the value function in conjunction with a DTA metric forms a barrier certificate to provide guarantees of scLTL satisfaction. To circumvent the computational demands of solving a sequence of HJB equations, we take a learning-based approach to approximately solve a sequence of reach-avoid optimal control problems online without requiring full knowledge of the system dynamics. Furthermore, we developed sufficient conditions for which this approximate solution approach maintains formal guarantees and supported these claims with a numerical example.   

A limitation of our approach that will be explored in future work is the inability to formally verify if the sufficient conditions of Theorem \ref{thm:barrier3} are satisfied for a given set of learning parameters. This technical challenge stems from the uniformly ultimately bounded convergence results of Proposition \ref{prop:UUB}, hence establishing an asymptotic convergence result, as in \cite{VamvoudakisTNNLS15}, may provide a pathway towards eliminating this complication. Despite this, we illustrated through numerical examples that this condition is not overly restrictive in practice. Similar to most Lyapunov-based approaches, the conditions for each result are only sufficient and thus our approach is not complete; however, our approach is sound in the sense that satisfaction of the sufficient conditions implies satisfaction of the given specification. Directions for future research include extensions to generic LTL formulas and decentralized systems where distributed teams of agents must act locally to satisfy global objectives.

\section{Acknowledgments}
This work was partially supported by the NSF under grants IIS-1723995, IIS-2024606,
and GRFP DGE-1840990. Any opinions, findings, and conclusions or
recommendations expressed in this material are those of the author(s) and do not necessarily reflect the views of the NSF. We thank the anonymous reviewers for their helpful comments and suggestions which have improved the quality of this paper.

\section*{Appendix}
The constants used in the proof of Proposition \ref{prop:UUB} are defined as follows:\\ $\varphi_{a}\triangleq \tfrac{3\sqrt{3}(k_{c1}+k_{c2})\overline{\|G_{\sigma}\|}\overline{\|W\|}}{64\sqrt{\gamma_1}} + \tfrac{\overline{\|\nabla_e W\|}\overline{\|G_R\nabla_e\sigma^T\|}}{2},\,
\varphi_{c\theta}\triangleq\tfrac{3\sqrt{3}(k_{c1}+k_{c2})\overline{\|\Delta_{F}\|}\big(\overline{\|Y\|} + \overline{\|gg^{\dag} \|}\overline{\|Y_d\|}\big)}{16\sqrt{\gamma_1}},\,
\varphi_{ac}\triangleq k_{a1} + \tfrac{3\sqrt{3}(k_{c1}+k_{c2})\overline{\|G_{\sigma}\| \|W\|}}{64\sqrt{\gamma_1}} + \tfrac{\overline{\|\nabla_e W\|\|G_R\nabla_e\sigma^T \|}}{2\underline{\Gamma}},\,
\iota_c\triangleq\tfrac{3\sqrt{3}(k_{c1}+k_{c2})\overline{\|\Delta\|}}{16\sqrt{\gamma_1}} + \tfrac{\overline{\|\nabla_e W\|}}{\underline{\Gamma}}\big( \overline{\|F\|} + \tfrac{\overline{\|G_R\nabla_e\sigma^TW \|}}{2} + \tfrac{\overline{\|G_R\nabla_e\bar{B}^T \|}}{2}\big),\,
\iota_a\triangleq k_{a2}\overline{\|W\|} + \tfrac{3\sqrt{3}(k_{c1}+k_{c2})\overline{\|G_{\sigma}\| \|W\|}^2}{64\sqrt{\gamma_1}} + \overline{\|\nabla_e W\|}\big(\overline{\|F\|}  + \tfrac{\overline{\|G_R\nabla_e\sigma^TW\|}}{2} +\tfrac{\overline{\|G_R\nabla_e\bar{B}^T\|}}{2}\big) + \tfrac{\overline{\|\nabla_e V^* \|}\overline{\|G_R\nabla_e\sigma^T\|}}{2},\,Y_d\triangleq Y(x_d)$.

\bibliographystyle{ieeetr}
\bibliography{ms}

\end{document}